\documentclass[conference,final]{IEEEtran}
\usepackage{amsfonts}
\usepackage{dsfont}
\usepackage{floatrow} 
\usepackage{setspace}
\usepackage{color}
\usepackage[table]{xcolor}
\usepackage{amssymb}
\usepackage{cite}
\usepackage[cmex10]{amsmath}
\usepackage{algorithm, algpseudocode}
\usepackage{algorithmicx}
\usepackage{array} 
\usepackage{mathrsfs}
\usepackage{graphicx}
\usepackage{latexsym}
\usepackage{amscd}
\usepackage{amsfonts}
\usepackage{caption}
\usepackage{subcaption}
\usepackage{amsmath,amscd,amssymb,verbatim}
\usepackage{hyperref}
\usepackage{graphics}
\usepackage{amsthm}
\usepackage[T1]{fontenc}
\usepackage[utf8]{inputenc}
\usepackage[short]{optidef}
\usepackage{mathtools}
\usepackage{soul}
\usepackage{textcomp}
\usepackage{xcolor}
\usepackage{datetime}
\usepackage{microtype}

\usepackage{blindtext}
\IEEEoverridecommandlockouts

\newtheorem{theorem}{Lemma}
\newtheorem{corollary}{Corollary}

\begin{document}
\renewcommand{\thefootnote}{\fnsymbol{footnote}}
\title{Joint Optimization for Security and Reliability in Round-Trip Transmissions for URLLC services 
%\thanks{
%Part of this work, i.e., the optimization problem in~\eqref{problem:single} with a single variable, was presented at 2022 IEEE International Conference on Communications~\cite{Zhu_ICC_2022}.   

 %Y. Zhu, X. Yuan, and Y. Hu are with School of Electronic Information, Wuhan University, 430072 Wuhan, China and   Chair of Information Theory and Data Analytics, RWTH Aachen University, 52074 Aachen, Germany. (Email:$zhu|yuan|hu$@inda.rwth-aachen.de). 

% B. Ai is with School of Electronic and Information Engineering, Beijing Jiaotong University, (Email: $boai$@bjtu.edu.cn).

 %B. Han is with the Division of Wireless Communications and Radio Positioning, RPTU Kaiserslautern-Landau, 67663 Kaiserslautern, Germany (Email: bin.han@rptu.de).

%R. Wang and A. Schmeink is with the Chair of Information Theory and Data Analytics, RWTH Aachen University, 52074 Aachen, Germany (Email: $wang|schmeink$@inda.rwth-aachen.de). 
% }
}
 
\author{
    \IEEEauthorblockN{Xinyan Le$^{1,2}$,  Yao Zhu$^{1,2}$, Yulin Hu$^{1,2}$, Bin~Han$^{3}$}
    \IEEEauthorblockA{$^{1}$ School of Electronic Information, Wuhan University, Wuhan 430000, China}
    \IEEEauthorblockA{$^{2}$ 6G Intelligent Connectivity International Science and Technology Cooperation Center, Hubei, China}
%     \IEEEauthorblockA{$^{3}$ Chair of Information Theory
% and Data Analytics, RWTH Aachen University,  Aachen, 52068, Germany}
\IEEEauthorblockA{$^{3}$ University of Kaiserslautern (RPTU), 67663, Germany}
    % \IEEEauthorblockA{\{zhangsan\}@XXX.com, \{lisi, wangwu\}@XXX.edu.cn, {g.li}@XXX.com}
    \IEEEauthorblockA{\{\href{mailto:2023202121358@whu.edu.cn}{Le}, \href{mailto: yao.zhu@whu.edu.cn}{Zhu}, \href{mailto:yulin.hu@whu.edu.cn}{Hu}\}@whu.edu.cn, %\{\href{mailto:anke.schmeink@isek.rwth-aachen.de}{Schmeink}\}@inda.rwth-aachen.de
    bin.han@rptu.de} % 给人名附上邮箱地址
}

\maketitle
\vspace{-5pt}\par
\begin{abstract}
Physical layer security (PLS) is a potential solution for secure and reliable transmissions in future Ultra-Reliable and Low-Latency Communications (URLLC). This work jointly optimizes redundant bits and blocklength allocation in practical round-trip transmission scenarios. To minimize the leakage-failure probability, a metric that jointly characterizes security and reliability in PLS, we formulate an optimization problem for allocating both redundant bits and blocklength. By deriving the boundaries of the feasible set, we obtain the globally optimal solution for this integer optimization problem. 
To achieve more computationally efficient solutions, we propose a block coordinate descent (BCD) method that exploits the partial convexity of the objective function. Subsequently, we develop a majorization-minimization (MM) algorithm through convex approximation of the objective function, which further improves computational efficiency. Finally, we validate the performance of the three proposed approaches through simulations, demonstrating their practical applicability for future URLLC services.

\end{abstract}
%\vspace*{-0.12cm}
\vspace{-7pt}\par

\section{Introduction}
\label{sec:intro}
\vspace{-2pt}\par

URLLC represents a key technology direction for 6G evolution, enabling real-time performance for Industrial Internet of Things (IIoT) applications such as precision remote control, multi-device collaborative operations, and smart factory scheduling. These scenarios typically involve massive device connectivity and round-trip communications, creating inherent security vulnerabilities. The heterogeneous nature of IIoT networks introduces significant attack surfaces, as malicious devices can exploit network access to eavesdrop on legitimate communications~\cite{V2X_security_intro_2019}. 
This security challenge is particularly critical for URLLC services, which rely on short-packet transmissions that cannot accommodate traditional cryptographic overheads while maintaining strict latency constraints.

PLS offers a promising solution to these challenges. Traditional PLS methods exploit the inherent randomness of wireless channels to create advantages for legitimate users over eavesdroppers. Wyner's seminal work~\cite{Wyner_wiretap_1975} demonstrated that perfect secrecy is achievable by adding redundancy at the transmitter, ensuring eavesdroppers are unable to decode the protected information. 
However, these classical approaches assume infinite blocklength, allowing arbitrary coding gains. This assumption fails for future URLLC systems, which must use short packets to meet strict latency requirements.

In short-packet transmissions, the analysis of PLS changes fundamentally. When the blocklength is finite (FBL regime), decoding errors become inevitable even when operating below Shannon capacity~\cite{Polyanskiy_2010}, requiring probabilistic rather than deterministic security guarantees. Yang et al.~\cite{Yang_wiretap_2019} derived tight bounds on achievable secrecy rates for FBL codes, catalyzing extensive research on rate-based FBL security metrics. However, these approaches poorly capture IIoT requirements for packet-level guarantees that jointly ensure both security and reliability.
Recognizing this gap, a novel metric, leakage failure probability (LFP) is introduced in~\cite{Zhu_LFP_2023} to quantify packet-level security-reliability performance, providing practical guidance for blocklength allocation. However, the optimization framework only considers a single-transmission  scenario. This simplified scenario cannot capture the complexities of practical URLLC applications, which predominantly feature bidirectional, round-trip communications. Extending LFP-based design to these multi-round scenarios presents unresolved challenges in modeling interdependencies between forward and backward transmissions while maintaining computational tractability.

Motivated by the above observations, this work addresses the critical challenge of jointly optimizing redundant bits and blocklength for secure and reliable round-trip transmissions in future URLLC scenarios. We extend the LFP metric to capture both security and reliability requirements in a unified round-trip framework. Our main contributions are:
\begin{itemize}
    \item \textbf{Problem formulation and optimal solution:} We extend the LFP metric to round-trip transmissions and formulate an optimization problem for joint redundancy and blocklength allocation. By deriving explicit boundaries on the feasible region based on reliability and leakage thresholds, we obtain the globally optimal solution via integer programming as a performance benchmark.
    
    \item \textbf{Efficient algorithms via convexity analysis:} We develop two computationally efficient approaches by relaxing integer constraints. First, we propose a BCD method exploiting the partial convexity in each variable. Second, we develop an MM algorithm by establishing joint convexity properties and applying convex approximation, achieving faster convergence.
    
    \item \textbf{Performance validation:} Simulations demonstrate that both proposed algorithms achieve near-optimal performance while significantly reducing computational complexity, making them practical for real-world URLLC applications.
\end{itemize}

The remainder of this paper is organized as follows. Section II constructs the system model and the objective function. Section III formulates the optimization problem and presents the solution method. Section IV evaluates the proposed scheme through numerical results. Section V provides concluding remarks.\vspace{-7pt}\par

\section{Preliminaries}
\label{sec:preliminaries}
\vspace{-8pt}\par
\subsection{System Model}
Consider a scenario, where a machine-type round-trip transmission occurs between a transmitter Alice and a legitimate receiver Bob, with an eavesdropper Eve attempting to intercept the messages. In the forward transmission (indexed by 1), Alice sends a data packet with blocklength \( m_{1} \) containing a message of \( d_{m,1} \) bits, which is received by both Bob and eavesdropped by Eve. After receiving, Bob responds with a backward transmission (indexed by 2) using blocklength \( m_{2} \) and a message of \( d_{m,2}\) bits, which is received by Alice and eavesdropped by Eve. This round-trip communication scenario can be modeled as two Wyner's wiretap channels~\cite{Wyner_wiretap_1975}, one for each transmission direction, where Eve attempts to intercept the messages in both directions. 

To prevent information leakage to Eve, redundant bits must be added according to physical layer security principles. Specifically, \( d_{r,1} \) redundant bits are added for the forward transmission, and \( d_{r,2} \) redundant bits are added for the backward transmission. Typically, the redundant bits differ between the forward and backward transmissions due to asymmetric channel conditions. 

Alice, Bob, and Eve are each equipped with a single antenna. We denote the channel gains by $h_{ab}$ and $h_{ae}$ for transmissions from Alice to Bob and Eve, respectively. Since round-trip channels are not symmetric, we denote the channel gains from Bob to Alice and Eve by $h_{ba}$ and $h_{be}$, respectively. The channel gains are modeled as $h_{i}=\sqrt{\xi_{i}} \hat{h}_{i}$ for $i\in\{ab,ae,ba,be\}$, where $\xi_{i}$ represents the large-scale pathloss, and $\hat{h}_{i}\sim \mathcal{N}(0,1)$ represent the independent and identically distributed small-scale fading. 
The received signal for each transmission link can be expressed as:
\vspace{-5pt}
\begin{equation}
    y_i=\sqrt{p}h_is+n_i, \\[-5pt]
\end{equation}
where $i\in\{ab,ae,ba,be\}$, $p$ is the transmit power, $s$ is the information signal, and $n_{i} \sim \mathcal{CN}(0,\sigma^2_{i})$ is the additive white Gaussian noise for link $i$. The corresponding signal-to-noise ratio (SNR) with perfect channel state information (CSI) is given by:
\vspace{-5pt}
\begin{equation}
\gamma_{i}=\frac{p |{h}_{i}|^2}{\sigma^2_{i}}, \quad i\in\{ab,ae,ba,be\}.\vspace{-4pt}
\end{equation}

\subsection{Physical Layer Security in the {FBL} Regime}\vspace{-2pt}\par
In our system model of round-trip transmission based on FBL, for Gaussian wiretap channels, the perfect secrecy rate is given by~\cite{Schaefer_PLS_overview_2016}:
\vspace{-5pt}
\begin{equation}
    C_{s,1}=C_{ab}-C_{ae},\quad C_{s,2}=C_{ba}-C_{be}.\\[-5pt]
\end{equation}
where $C_{i}=\log_2(1+\gamma_{i})$ are the Shannon capacities of the respective links with $i\in \{ab,ae,ba,be\}$, and $C_{s,1}$ and $C_{s,2}$ represent the perfect secrecy rates for the forward and backward transmissions, respectively. However, in the FBL regime, the achievable secrecy rate with a given error probability of $\bar{\varepsilon}$ and a given information leakage probability $\bar{\delta}$ is given by:
\vspace{-5pt}
\begin{equation}
\begin{aligned}
        r^*_{s,j}=C_{s,j}-\sqrt{\frac{V(\gamma_{b,j})}{m_j}}Q^{-1}(\bar{\varepsilon}_{j})
    -\sqrt{\frac{V(\gamma_{e,j})}{m_j}}Q^{-1}(\bar{\delta}),\\[-5pt]
    \end{aligned}
\end{equation}
where $j \in \{1,2\}$ denotes the forward or backward transmission with $\gamma_{b,1}=\gamma_{ab}$, $\gamma_{e,1}=\gamma_{ae}$, $\gamma_{b,2}=\gamma_{ba}$, and $\gamma_{e,2}=\gamma_{be}$. Here, $Q^{-1}(\cdot)$ is the inverse Q-function with $Q(x)= \int_{x}^{\infty}{\frac{1}{\sqrt{2\pi}}e^{-{\frac{t^2}{2}}}\mathrm{d}t}$, and $V(\gamma)$ is the channel dispersion given by $V(\gamma)=1-(1+\gamma)^{-2}$.
%While the achievable secrecy rate characterizes the average PLS performance for each transmission direction, it may not adequately capture the requirements of our scenario. In practical URLLC applications involving round-trip communications, we are concerned with the performance of a single round of transmissions rather than long-term average behavior. Therefore, we need a metric that jointly characterizes both security and reliability for the whole round of transmissions, which motivates our focus on packet-level performance metrics.\vspace{-3pt}\par
While the achievable secrecy rate characterizes average PLS performance, it inadequately captures practical URLLC requirements. Round-trip applications demand packet-level metrics that jointly characterize security and reliability for individual transmission rounds rather than long-term averages.\vspace{-3pt}\par

\section{Optimization and Analysis}
\label{sec:system analyis}

\subsection{LFP Characterization and Problem Formulation}
\allowdisplaybreaks

To characterize the packet-level security and reliability performance, we first express the achievable secrecy rate for each transmission direction in terms of individual achievable transmission rates. Following~\cite{Yang_wiretap_2019}, the achievable secrecy rate for the forward transmission (index $j=1$) can be decomposed as:
\vspace{-6pt}
\begin{equation}
\begin{aligned}
r^*_{s,1}&=r^*_{ab}-r^*_{ae},\\[-7pt]
\end{aligned}
\end{equation}
where $r^*_{ab}$ and $r^*_{ae}$ represent the achievable transmission rates from Alice to Bob and from Alice to Eve, respectively. Similarly, for the backward transmission (index $j=2$), it holds that $r^*_{s,2}=r^*_{ba}-r^*_{be}$.
% \vspace{-6pt}
% \begin{equation}
% \begin{aligned}
% r^*_{s,2}&=r^*_{ba}-r^*_{be}.\\[-7pt]
% \end{aligned}
% \end{equation}

For each transmission direction $j\in\{1,2\}$, the total number of transmitted bits consists of the message bits $d_{m,j}$ and the redundant bits $d_{r,j}$ for security, i.e., 
\vspace{-6pt}
\begin{equation}
d_j=d_{m,j}+d_{r,j}.\\[-5pt]
\end{equation}
Given the blocklength $m_j$ and the SNR $\gamma_i$ for each link $i\in\{ab,ae,ba,be\}$, the error probability (i.e., the probability that the receiver fails to decode correctly) can be expressed as~\cite{Zhu_LFP_2023}:
\vspace{-7pt}
\begin{equation}
\varepsilon_{i}=Q\left(\sqrt{\frac{m_j}{V(\gamma_{i})}}\left(C(\gamma_{i})-\frac{d_j}{m_j}\right)\ln 2\right).\\[-5pt]
\end{equation}

We define the event $X_i$ as the successful decoding event at the corresponding receiver for link $i$. The probability of successful decoding is $P(X_i=1)=1-\varepsilon_i$, and the error probability is $P(X_i=0)=\varepsilon_i$. A secure and reliable transmission in direction $j$ occurs when the legitimate receiver decodes successfully while the eavesdropper fails to decode. We denote this event by $Z_j$. For the forward transmission ($j=1$), we have:
\vspace{-7pt}
\begin{equation}
P(Z_{1}=1) = P(X_{ab}=1 \cap X_{ae}=0)=(1-\varepsilon_{ab})\varepsilon_{ae}.\\[-7pt]
\end{equation}
Similarly, for the backward transmission ($j=2$):
\vspace{-7pt}
\begin{equation}
P(Z_{2}=1) = P(X_{ba}=1 \cap X_{be}=0)=(1-\varepsilon_{ba})\varepsilon_{be}.\\[-7pt]
\end{equation}

A successful round-trip transmission, denoted by event $Y$, requires both the forward and backward transmissions to be secure and reliable simultaneously, i.e., $Y=Z_1\cap Z_2$. The LFP, defined as the probability that the round-trip transmission fails to be both secure and reliable~\cite{Zhu_LFP_2023}, is given by:
\vspace{-6pt}
\begin{equation}
 \varepsilon_{LF}=1-P(Y=1)=1-(1-\varepsilon_{ab})\varepsilon_{ae}(1-\varepsilon_{ba})\varepsilon_{be}.\\[-6pt]
\end{equation}

Our objective is to minimize the LFP $\varepsilon_{LF}$ by jointly optimizing the redundant bits $d_{r,1}$ and $d_{r,2}$, as well as the blocklength allocations $m_1$ and $m_2$ for the round-trip transmission. We formulate the following integer programming problem:
\vspace{-5pt}
\begin{mini!}[2]
    {_{m_{1},m_{2},{d_{r,1}},{d_{r,2}}}}{\varepsilon_{LF}}
    {\label{problem:single}}{(\mathrm{OP}):}
\addConstraint{\varepsilon_{ab}\leq\varepsilon_{ab}^{\max}, \varepsilon_{ba}\leq\varepsilon_{ba}^{\max} }\label{con:op2}
\addConstraint{\varepsilon_{ae}\geq\varepsilon_{e}^{\max},~ \varepsilon_{be}\geq\varepsilon_{e}^{\max}}\label{con:op3}
    \addConstraint{m_{1}+m_{2}\leq M}\label{con:op4}
    \addConstraint{d_{r,1},d_{r,2}\in \mathbb{N},~ m_{1},m_{2}\in \mathbb{N}^{+}.}\label{con:op5}
\end{mini!}\vspace{-15pt}\par
Constraints~\eqref{con:op2} and~\eqref{con:op3} specify the reliability and leakage thresholds, respectively. Constraint~\eqref{con:op4} limits the total blocklength for the round-trip communication to satisfy the delay requirement, where $M$ is the maximum allowable blocklength. Constraint~\eqref{con:op5} enforces the integer constraints on the optimization variables.

The above problem is an integer optimization problem, which can be solved with integer programming if the feasible set is bounded. 
According to Lemma 1 in~\cite{joint_2024}, the objective function $\varepsilon_{LF}$ is monotonically increasing in $d_{r,1}$ and $d_{r,2}$. Thus, the reliability and leakage thresholds define the feasible range of $d_{r,1}$ and $d_{r,2}$. Therefore, constraints~\eqref{con:op2} and~\eqref{con:op3} can be reformulated as equivalent constraints:

\vspace{-8pt}
\begin{equation}
\left\{
\begin{aligned}
{d_{r,j}} \!&\leq m_j C(\gamma_{b,j}) \!- \!\sqrt{m_j V(\gamma_{b,j})} \ Q^{-1}(\varepsilon_{b,j}^{\max})\!=\!d_{r,j}^{\max}, \\[-4pt]
{d_{r,j}}\! &\geq m_j C(\gamma_{e,j})\! - \!\sqrt{m_j V(\gamma_{e,j})} \ Q^{-1}(\varepsilon_{e}^{\max})\!=\!d_{r,j}^{\min},\\[-8pt]
\end{aligned}
\right.
\end{equation}
where $j \in \{1,2\}$ represents the transmission direction. Thus, (OP) is transformed into an integer optimization problem with a bounded feasible set, which can be solved using integer programming with complexity $O(d_{r,1}^{\max} d_{r,2}^{\max} M^2)$. We can rephrase this problem as:
\begingroup  % 局部生效
\setlength{\abovedisplayskip}{0pt}
\setlength{\belowdisplayskip}{0pt}
%\vspace{-10pt}
\begin{mini!}[2]
    {_{m_{1},m_{2},{d_{r,1}},{d_{r,2}}}}{\varepsilon_{LF}}
    {\label{problem:once}}{(\mathrm{SP1}):}
    \addConstraint{d_{r,1}^{\min}\leq d_{r,1} \leq d_{r,1}^{\max}}\label{con:SP11}
    \addConstraint{d_{r,2}^{\min}\leq d_{r,2} \leq d_{r,2}^{\max}}\label{con:SP12}
    \addConstraint{~\eqref{con:op4}, ~\eqref{con:op5}}\vspace{-5pt}
\end{mini!}
\endgroup
\vspace{-13pt}

where constraints~\eqref{con:SP11} and~\eqref{con:SP12} are the constraints of $d_{r,1},d_{r,2}$ that specified by the leakage thresholds.
\vspace{-5pt}\par
\subsection{Dimensionality Reduction} 
While integer programming provides the globally optimal solution, it suffers from high computational complexity and lacks insights for system design. We therefore develop more efficient approaches by first analyzing the monotonicity properties of $\varepsilon_{LF}$ to reduce the problem from four to three variables. We begin with the following lemma:\vspace{-3pt}
\begin{theorem}
    \label{lemma:start_time0}
Within the feasible set of (SP1), $\varepsilon_{be}$ and $\varepsilon_{ba}$ are decreasing in $m_2$ and increasing in $d_{r,2}$.
\end{theorem}
\begin{proof}
To prove this theorem, we compute the partial derivatives of $\varepsilon_{be}$ and $\varepsilon_{ba}$ with respect to $m_2$ and $d_{r,2}$, respectively.
\vspace{-7pt}
\begin{equation}
\begin{aligned}
\frac{\partial \varepsilon_{i}}{\partial m_2}
&= \ln 2 \frac{\partial \varepsilon_{i}}{\partial w_{i}} \left( \frac{C(\gamma_{i}) - \frac{d_2}{m_2}}{2 \sqrt{m_2 V(\gamma_{i})}} + \frac{d_2}{m_2^2} \sqrt{\frac{m_2}{V(\gamma_{i})}} \right) \leq 0, \\[-3pt]
\end{aligned}
\end{equation}
\vspace{-3pt}
\begin{equation}
\begin{aligned}
\frac{\partial \varepsilon_{i}}{\partial d_{r,2}} 
&= -\ln 2 \frac{\partial \varepsilon_{i}}{\partial w_{i}} \frac{\ln 2}{\sqrt{m_2 V(\gamma_{i})}} \geq 0, \\[-7pt]
\end{aligned}
\end{equation}
\vspace{-3pt}
where $w_i=-(m_2/V(\gamma_i))^{1/2}((C(\gamma_i)-{d_2}/{m_2})\ln 2)$ is an auxiliary function with $i \in \{ba,be\}$.
\end{proof}
\vspace{-2pt}\par
Lemma~\ref{lemma:start_time0} establishes the monotonicity properties of the error probabilities in the backward transmission. These properties are crucial for characterizing the behavior of the objective function. Building upon this lemma, we establish the following corollary that enables variable reduction:\vspace{-3pt}
\begin{corollary}
    \label{lemma:start_time1}
   For the feasible set (SP1), the objective function $\varepsilon_{LF}$ attains its minimum value only when the equality condition $m_1+m_2 \leq M$ holds, i.e., when $m_1+m_2=M$.
\end{corollary}
\begin{proof}
We can fix the variables $m_1,d_{r,1},d_{r,2}$. If only the variable $m_2$ is considered, where $m_2 \in (M - m_1 - 1 ,M)$, $\varepsilon_{LF}$ can be written as:
\vspace{-8pt}
\begin{equation}
\varepsilon_{LF} = 1 - k(1 - \varepsilon_{{ba}}) \varepsilon_{be}\\[-5pt]
\end{equation}
where $k=(1-\varepsilon_{{ab}})\varepsilon_{ae}$ can be regarded as a constant in this proof.
Then, to prove that $\varepsilon_{LF}$ when $m_1 + m_2 = M$, we can prove that $\varepsilon_{LF}$ is monotonically increasing with $m_2$.

According to Lemma~\ref{lemma:start_time0} %gives the expressions of $\frac{\partial \varepsilon_{{ba}}}{\partial m_2}$ and $\frac{\partial \varepsilon_{be}}{\partial m_2}$, from which we know that 
both $\varepsilon_{{ba}}$ and $\varepsilon_{be}$ are decreasing in $m_2$. Since $\gamma_{be} \neq \gamma_{{ba}}$, their decreasing rates are different. The partial derivative of $\varepsilon_{LF}$ in $m_2$ is given by :
\vspace{-5pt}
\begin{equation}
    \frac{\partial \varepsilon_{LF}}{\partial m_2} = k \varepsilon_{be} \frac{\partial \varepsilon_{{ba}}}{\partial w_{{ba}}} \cdot \frac{\partial w_{{ba}}}{\partial m_2} + k (\varepsilon_{{ba}} - 1) \frac{\partial \varepsilon_{be}}{\partial w_{be}} \cdot \frac{\partial w_{be}}{\partial m_2}\\[-5pt]
\end{equation}
Next, we introduce $r_2 = \frac{{d_{r,2}} + d_{m_2}}{m_2}$~\cite{joint_2024}. Since both ${d_{r,2}}$ and $m_2$ are variables, $r_2$ can be made to satisfy $r_2 \in N_+$.

Thus, for any $m_2$, there exists a corresponding $r_2^*$ that satisfies :
\vspace{-5pt}
\begin{equation}
\begin{aligned}
\frac{\partial \varepsilon_{LF}}{\partial m_2} &= k \varepsilon_{be} \frac{\partial \varepsilon_{{ba}}}{\partial w_{{ba}}} \cdot \frac{c(\gamma_{{ba}}) - r_2^* \ln 2}{2\sqrt{m_2 V(\gamma_{{ba}})}} \\[-1pt]
&+ k(\varepsilon_{{ba}} - 1) \frac{\partial \varepsilon_{be}}{\partial w_{be}} \cdot \frac{c(\gamma_{be}) - r_2^* \ln 2}{2\sqrt{m_2 V(\gamma_{be})}} \leq 0\\[-3pt]
\end{aligned}
\end{equation}
So that $\varepsilon_{LF}$ is decreasing in $m_2$. That is, whenever the case $m_1^* + m_2^* < M$ occurs, there always exists $m_2 > m_2^*$  such that $\varepsilon_{LF} < \varepsilon_{LF}^*$. Therefore, $\varepsilon_{LF}$ attains its minimum value only when the equality $m_1 + m_2 = M$ holds.
\end{proof}\vspace{-5pt}\par
Corollary~\ref{lemma:start_time1} reveals that the optimal solution always utilizes the full available blocklength, allowing us to express \( m_2 \) in terms of \( m_1 \) as \( m_2 = M - m_1 \). This analytical finding reduces the dimensionality of the optimization problem from four variables to three, and (SP1) can be reformulated as:
\vspace{-5pt}
\begin{mini!}[2]
    {_{m_{1},d_{r,1},d_{r,2}}}{\varepsilon_{LF}}
    {\label{problem:two}}{(\mathrm{SP2}):}
    \addConstraint{d_{r,1}^{\min}\leq d_{r,1} \leq d_{r,1}^{\max}}\label{con:SP21}
    \addConstraint{d_{r,2}^{\min}\leq d_{r,2} \leq d_{r,2}^{\max}}\label{con:SP22}
    \addConstraint{1 \leq m_{1}\leq M}\label{con:SP23}
    \addConstraint{d_{r,1},d_{r,2}\in \mathbb{N},~ m_{1}\in \mathbb{N}^{+}}\label{con:SP24}
\end{mini!}\vspace{-18pt}\par
where constraints~\eqref{con:SP21} and~\eqref{con:SP22} are the leakage thresholds of $d_{r,1}$ and $d_{r,2}$. Note that constraint~\eqref{con:SP23} is the range of \( m_1 \) after substituting \( m_2 = M - m_1 \). \vspace{-3pt}\par
 \vspace{-5pt}\par
\subsection{Partial Convexity Analysis and Block Coordinate Descent}
To further analyze the convexity properties of the problem, we relax the integer variables to real-valued ones. This relaxation allows us to  to characterize the convexity properties of the objective function. Accordingly, we reformulate (SP2) by relaxing the integer constraints:
\vspace{-7pt}
\begin{mini!}[2]
{_{m_{1},{d_{r,1}},{d_{r,2}}}}{\varepsilon_{LF}}
{\label{problem:three}}{(\mathrm{SP3}):}
\addConstraint{~\eqref{con:SP22},~\eqref{con:SP23},~\eqref{con:SP24}}\label{con:SP31}
\addConstraint{d_{r,1},d_{r,2}\in \mathbb{R}^+,~ m_{1}\in \mathbb{R}^{++}}\label{con:SP32}
\end{mini!}\vspace{-18pt}\par
where constraint~\eqref{con:SP32} relaxes the integer constraints in~\eqref{con:SP24} to real variables.

We now establish the partial convexity of the objective function $\varepsilon_{LF}$ with respect to each variable. These results lay the foundation for an efficient block coordinate descent (BCD) solution approach. We begin with the following theorem:\vspace{-5pt}\par
\begin{theorem}
   With the feasible set of (SP3), the objective function
$\varepsilon_{LF}$ is convex in blocklength $m_1$.
    \label{lemma:start_time1}
\end{theorem}\vspace{-9pt}\par
\begin{proof}
According to lemma 5 of~\cite{zhu_2025}, $\varepsilon_{{ab}}$ and $\varepsilon_{{ba}}$ are convex in $m_1$, while $\varepsilon_{ae}$ and $\varepsilon_{be}$ are concave in $m_1$ .So that,$\frac{\partial^2{\varepsilon_{{ab}}}}{\partial m_1^2} \geq 0,$ $\frac{\partial^2{\varepsilon_{{ba}}}}{\partial m_1^2} \geq 0.$ Thus, we can prove that:\vspace{-7pt}
\begin{equation}
\begin{aligned}
\frac{\partial^{2}\varepsilon_{LF}}{\partial m_{1}^{2}}  
&=-{P}\bigg[\underbrace{(-q_{1}+q_{2}-q_{3}+q_{4})^{2}-(q_{1}^{2}+q_{2}^{2}+q_{3}^{2}+q_{4}^{2})}_{\leq 0} \\[-3pt]
 & \qquad \!-\!\underbrace{\frac{\frac{\partial^{2}\varepsilon_{{ab}}}{\partial m_1^{2}}}{1-\varepsilon_{{ab}}}}_{\geq 0}\!+\!\underbrace{\frac{\frac{\partial^{2}\varepsilon_{ae}}{\partial {m_1}^{2}}}{\varepsilon_{ae}}}_{\leq 0}\!-\!\underbrace{\frac{\frac{\partial^{2}\varepsilon_{{ba}}}{\partial {m_1}^{2}}}{1-\varepsilon_{{ba}}}}_{\geq 0}\!+\!\underbrace{\frac{\frac{\partial^{2}\varepsilon_{be}}{\partial {m_1}^{2}}}{\varepsilon_{be}}}_{\leq 0}\bigg] \geq 0\\[-8pt]
\end{aligned}
\end{equation}
where $P=1-\varepsilon_{LF}$, $q_1=\frac{\frac{\partial\varepsilon_{ab}}{\partial {m_1}}}{1-\varepsilon_{ab}}$, $q_2=\frac{\frac{\partial\varepsilon_{ae}}{\partial {m_1}}}{\varepsilon_{ae}}$, $q_3=\frac{\frac{\partial\varepsilon_{ba}}{\partial {m_1}}}{1-\varepsilon_{ba}}$, and $q_4=\frac{\frac{\partial\varepsilon_{be}}{\partial {m_1}}}{1-\varepsilon_{be}}$ are auxiliary functions. Note that according to Theorem~\ref{lemma:start_time0}, we can infer that $\frac{\partial{\varepsilon_{ab}}}{\partial m_1} \geq 0$ and $\frac{\partial{\varepsilon_{ba}}}{\partial m_1} \geq 0$. Since they are all probability functions, we have $0 \leq \varepsilon_i \leq 1$ for $i\in\{ab,ae,ba,be\}$.
\end{proof}\vspace{-7pt}\par
\begin{corollary}
   With the feasible set of (SP3), the objective function $\varepsilon_{LF}$ is convex in $d_{r,1}$ and $d_{r,2}$, respectively.
    \label{corollary:start_time2}
\end{corollary}\vspace{-7pt}\par
\begin{proof}
We first prove that $\varepsilon_{ab}$ is convex in $d_{r,1}$ and $\varepsilon_{ba}$ is convex in $d_{r,2}$:\vspace{-8pt}
\begin{equation}
\begin{aligned}
\frac{\partial^{2} \varepsilon_{ab}}{\partial d_{r,1}^{2}} 
&\!=\! -\!\left(\!\!\frac{\ln 2}{\sqrt{m_{1} V(\gamma_{ab})}}\!\right)^{2}  \underbrace{\!\!\!\chi_{ab}\left(d_{r,1}\right)}_{\geq 0}  \underbrace{\frac{\partial \varepsilon_{ab}}{\partial \chi_{ab}\left(d_{r,1}\right)}}_{\geq 0}\geq 0,\\[-7pt]
\end{aligned}
\end{equation}
where $\chi_{i}\left(d_{r,i}\right) = \sqrt{\frac{m_1}{V(\gamma_i)}}\left((C(\gamma_i) - \frac{d_1}{m_1})\ln 2\right)$ is an auxiliary function with $i \in \{ab,ae\}$. According to the normal distribution probability function, we have $\chi_{ab}\left(d_{r,1}\right) \geq 0$ and $\chi_{ae}\left(d_{r,1}\right) \leq 0$. 
The second derivative $\frac{\partial^{2} \varepsilon_{ba}}{\partial d_{r,2}^{2}}$ differs from $\frac{\partial^{2} \varepsilon_{ab}}{\partial d_{r,1}^{2}}$ only in the constant term. From Theorem~\ref{lemma:start_time0}, we have
$
{\partial^{2} \varepsilon_{ba}}/{\partial d_{r,2}^{2}} \geqslant 0.
$ Next, we prove that $\varepsilon_{ae}$ is concave in $d_{r,1}$ and $\varepsilon_{be}$ is concave in $d_{r,2}$:
\begin{equation}
\frac{\partial^{2} \varepsilon_{ae}}{\partial d_{r,1}^{2}} \!=   \!- \!\left( \! \!\frac{\ln 2}{\sqrt{m_1 V\left(\gamma_{ae}\right)}}\right)^{2}  \underbrace{ \! \! \!\chi_{ae}\left(d_{r,1}\right) }_{\leq 0}\underbrace{\frac{\partial \varepsilon_{ae}}{\partial \chi_{ae}\left(d_{r,1}\right)}}_{\geq 0}  \!\leq0.\\[-9pt]
\end{equation}
Following the same derivation, we obtain $\frac{\partial^{2} \varepsilon_{be}}{\partial d_{r,2}^{2}} \leq 0$.
We now analyze the behavior of $\varepsilon_{LF}$ with respect to $d_{r,1}$ and $d_{r,2}$:
\vspace{-5pt}
\begin{equation}
\begin{aligned}
\frac{\partial^{2} \varepsilon_{LF}}{\partial d_{r,1}^{2}} &=  \left(1-\varepsilon_{ba}\right) \varepsilon_{e,2}
\bigg[\underbrace{\frac{\partial^{2} \varepsilon_{ab}}{\partial d_{r,1}^{2}} \varepsilon_{e,1}}_{\geq 0}+ \underbrace{2 \frac{\partial \varepsilon_{ab}}{\partial d_{r,1}} \frac{\partial \varepsilon_{e,1}}{\partial d_{r,1}}}_{\geq 0}\\[-7pt]
&- \underbrace{\left(1-\varepsilon_{ab}\right) \frac{\partial^{2} \varepsilon_{e,1}}{\partial d_{r,1}^{2}}}_{\leq 0}\bigg] \geq 0.\\[-7pt]
\end{aligned}
\end{equation}
Similarly, we have $\frac{\partial^{2} \varepsilon_{LF}}{\partial d_{r,2}^{2}} \geq 0$.
\end{proof}\vspace{-7pt}\par
Lemma~\ref{lemma:start_time1} and Corollary~\ref{corollary:start_time2} establish the partial convexity of $\varepsilon_{LF}$ with respect to $m_1$, $d_{r,1}$, and $d_{r,2}$ individually. These analytical results enable us to solve (SP3) efficiently using BCD. The BCD algorithm iteratively optimizes one variable while fixing the others: we first fix $d_{r,1}$ and $d_{r,2}$ to solve a convex optimization problem in $m_1$, then fix $m_1$ and $d_{r,2}$ to solve for $d_{r,1}$, and finally fix $m_1$ and $d_{r,1}$ to solve for $d_{r,2}$. This iteration continues until convergence. The computational complexity is $O(k(m_{1}^{\max}+d_{r,1}^{\max}+d_{r,2}^{\max}))$, where $k$ is the number of iterations. 
Since the variables of original problem are integer, the solution obtained from the relaxed problem (SP3) yields real values. To reconstruct the integer solution, we compare the objective function at the integer neighbors of the relaxed solution and select the values that minimize the objective.\vspace{-7pt}\par

\subsection{Convex Approximation and Majorization-Minimization Algorithm}
While the BCD algorithm reduces computational complexity compared to integer programming, its iteration-level complexity remains proportional to the number of variables. With three optimization variables in (SP3), the overall complexity may still be prohibitive for practical implementation. To address this limitation, we leverage convex approximation techniques that enable the use of the majorization-minimization (MM) algorithm, which offers improved convergence properties and lower per-iteration complexity.

The key challenge in applying convex approximation lies in the multiplicative structure of $\varepsilon_{LF}$, which involves the product $(1-\varepsilon_{{ab}})\varepsilon_{ae}(1-\varepsilon_{{ba}})\varepsilon_{be}$. To facilitate convex analysis, we perform an objective function transformation. Note that minimizing $\varepsilon_{LF} = 1-(1-\varepsilon_{{ab}})\varepsilon_{ae}(1-\varepsilon_{{ba}})\varepsilon_{be}$ is equivalent to maximizing $(1-\varepsilon_{{ab}})\varepsilon_{ae}(1-\varepsilon_{{ba}})\varepsilon_{be}$. Since all error probabilities satisfy $\varepsilon \in (0,1)$, the maximization problem can be equivalently transformed into minimizing the reciprocal $\frac{1}{(1-\varepsilon_{{ab}})\varepsilon_{ae}(1-\varepsilon_{{ba}})\varepsilon_{be}}$. Accordingly, we reformulate the optimization problem as:
\vspace{-5pt}
\begin{mini!}[2]
    {_{m_{1}\!,d_{r,1}\!,d_{r,2}}}{f(m_1,d_{r,1},d_{r,2}) \!=\! \frac{1}{(1\!-\!\varepsilon_{{ab}})\varepsilon_{ae}(1\!-\!\varepsilon_{{ba}})\varepsilon_{be}}}
    {\label{problem:fourth}}{(\mathrm{SP4})\!:\!}
    %\addConstraint{\bar{\Delta}_i\leq \Delta_{\max},\forall i\in\mathcal{I}}
    \addConstraint{\eqref{con:SP31},  \eqref{con:SP32}}\label{con:SP41}
\end{mini!}\vspace{-18pt}\par
Note that at convergence, a local minimum of (SP4) corresponds to a local minimum of the original $\varepsilon_{LF}$. This transformation converts the multiplicative structure into a form amenable to convex approximation. To enable the MM algorithm, we establish the following upper bound based on the arithmetic-geometric mean inequality:
\begin{theorem}
    \label{lemma:start_time22}
Within the feasible set of (SP4), the function $f(m_1,d_{r,1},d_{r,2})$ can be upper-bounded by
\vspace{-8pt}
\begin{equation}
g(m_1,d_{r,1},d_{r,2})=\left( \frac{\frac{1}{1-\varepsilon_{{ab}}} + \frac{1}{\varepsilon_{ae}} + \frac{1}{1-\varepsilon_{{ba}}} + \frac{1}{\varepsilon_{be}}}{4} \right)^2\\[-7pt] 
\end{equation}
\end{theorem}
\begin{proof}
    This can be proven by directly apply the arithmetic-geometric mean inquality with the conditions that all error probabilities satisfy $0 < \varepsilon_i < 1$.
% By applying the arithmetic-geometric mean inequality, we have
% \vspace{-8pt}
% \begin{equation}
% \frac{1}{(1\!-\!\varepsilon_{{ab}})\varepsilon_{ae}(1\!-\!\varepsilon_{{ba}})\varepsilon_{be}} \!\!\leq\!\! \left( \!\!\frac{\frac{1}{1-\varepsilon_{{ab}}} \!+\! \frac{1}{\varepsilon_{ae}}\! + \!\frac{1}{1-\varepsilon_{{ba}}}\! +\! \frac{1}{\varepsilon_{be}}}{4} \!\!\right)^2 \\[-7pt]
% \end{equation}
%which holds since all error probabilities satisfy $0 < \varepsilon_{{ab}}, \varepsilon_{ae}, \varepsilon_{{ba}}, \varepsilon_{be} < 1$.
\end{proof}\vspace{-5pt}\par
This upper bound eliminates the multiplicative coupling between variables, thereby enabling efficient convex optimization. We next establish the joint convexity property that is essential for the MM algorithm:
\begin{theorem}
    \label{lemma:start_time4}
Within the feasible set of (SP4), the objective function $g(m_1,{d_{r,1}},{d_{r,2}})$ is jointly convex in ${d_{r,1}}$ and ${d_{r,2}}$.
\end{theorem}\vspace{-7pt}\par
\begin{proof}
First, we need to prove that the function is convex in ${d_{r,1}},{d_{r,2}}$, respectively.
\vspace{-7pt}
\begin{equation}
\begin{aligned}
\frac{\partial^2 g}{\partial d_{r_j}^2} & = \frac{1}{8} \bigg[ 
\left( \frac{\partial \frac{1}{1 - \varepsilon_{ab}} }{\partial r_j} + \frac{\partial \frac{1}{\varepsilon_{ae}}}{\partial r_j} \right)^2 + \left( \frac{1}{1 - \varepsilon_{ab}} + \frac{1}{\varepsilon_{ae}}\right. \\[-5pt]
& \left.  + \frac{1}{1 - \varepsilon_{ba}} + \frac{1}{\varepsilon_{be}} \right) \left( \frac{\partial^2\frac{1}{1 - \varepsilon_{ab}} }{\partial r_j^2} + \frac{\partial^2 \frac{1}{\varepsilon_{ae}}}{\partial r_j^2} \right) \bigg] \geq 0\\[-7pt]
\end{aligned}
\end{equation}
where $j\in(1,2)$. Next, we construct the Hessian matrix of this function in $d_{r,1}$ and $d_{r,2}$ to prove it has joint convexity, i.e., to prove that the Hessian matrix is positive semi-definite:
\vspace{-5pt}
\begin{equation}
H = \begin{pmatrix}
\frac{\partial^2 g}{\partial {d_{r,1}}^2} & \frac{\partial^2 g}{\partial {d_{r,1}} \partial {d_{r,2}}} \\
\frac{\partial^2 g}{\partial {d_{r,1}} \partial {d_{r,2}}} & \frac{\partial^2 g}{\partial {d_{r,2}}^2}
\end{pmatrix}
\vspace{-5pt}
\end{equation}
\begin{equation}
\begin{aligned}
\| H \| 
&= \frac{1}{64}\! (A \!+ \!B \!+ \!C\! + \!D)\! \left[ \!\left( \!\frac{\partial A}{\partial {d_{r,1}}}\! +\! \frac{\partial B}{\partial {d_{r,1}}} \!\right)\! \left(\!\frac{\partial C}{\partial {d_{r,2}}}\! + \!\frac{\partial D}{\partial {d_{r,1}}}\! \right)^2 \right. \\[-5pt]
&+\! \left(\! \frac{\partial C}{\partial {d_{r,1}}}\! + \!\frac{\partial D}{\partial {d_{r,1}}} \!\right)\! \left(\! \frac{\partial A}{\partial {d_{r,1}}}\! + \frac{\partial B}{\partial {d_{r,1}}} \!\right) ^2 \!\!\!\!+ \!(A\! +\! B+ \!C \\[-5pt]
& +\! D)\! \left(\! \frac{\partial^2 A}{\partial {d_{r,1}}^2} \!+\! \frac{\partial^2 B}{\partial {d_{r,1}}^2} \!\right)\! \!\left( \!\frac{\partial^2 C}{\partial {d_{r,2}}^2} + \frac{\partial^2 D}{\partial {d_{r,2}}^2}\! \right)\! \Bigg] \!\geq\!0\\[-7pt]
\end{aligned}
\end{equation}
where $A = \frac{1}{1 - \varepsilon_{{ab}}}$, $B = \frac{1}{\varepsilon_{ae}}$, $C = \frac{1}{1 - \varepsilon_{{ba}}}$, and $D = \frac{1}{\varepsilon_{be}}$ are auxiliary functions. It is not difficult to find that all leading principal minors of the Hessian matrix are non-negative; thus, we conclude that $g$ is jointly convex in ${d_{r,1}}$ and ${d_{r,2}}$.
\end{proof}\vspace{-6pt}
The joint convexity of ${d_{r,1}}$ and ${d_{r,2}}$ established above enables the application of the MM algorithm. Based on Theorems~\ref{lemma:start_time22} and ~\ref{lemma:start_time4}, we formulate the MM algorithm by constructing a surrogate function at a given local point $(m,\hat{d}_{r,1},\hat{d}_{r,2})$. Specifically, the problem in (SP4) can be reformulated as:
\vspace{-7pt}
\begin{mini!}[2]
    {_{m_{1},d_{r,1},d_{r,2}}}{g(m_1,d_{r,1},d_{r,2}|\hat d_{r,1},\hat d_{r,2})}
    {\label{problem:fifth}}{(\mathrm{SP5}):}
     \addConstraint{\eqref{con:SP31}, \eqref{con:SP32}}
\end{mini!}\vspace{-17pt}\par
The MM algorithm operates by iteratively minimizing the surrogate function $g$ instead of the original function $f$. In each iteration, we treat $m_1$ as fixed and optimize the block coordinates $(d_{r,1}, d_{r,2})$ jointly using the convex surrogate. The convergence of this approach is guaranteed by the following majorization chain:
\vspace{-7pt}
\begin{equation}
\begin{aligned}
f\!\left(d_{r,1}^{(k)}, d_{r,2}^{(k)}\right) &= g^{(k)}\left( d_{r,1}^{(k)}, d_{r,2}^{(k)}\right) \geq g^{(k)}\left( d_{r,1,\text{opt}}^{(k)}, d_{r,2,\text{opt}}^{(k)}\right) \\[-3pt]
&  \geq \! f^{(k+1)}\!\left(\!d_{r,1}^{(k+1)}, d_{r,2}^{(k+1)}\!\right)\! =\! g^{(k+1)}\!\left(\!d_{r,1}^{(k+1)}, d_{r,2}^{(k+1)}\!\right) \\[-3pt]
&\geq g^{(k+1)}\left(m_{1}, d_{r,1,\text{opt}}^{(k+1)}, d_{r,2,\text{opt}}^{(k+1)}\right)\\[-7pt]
\end{aligned}
\end{equation}

The overall algorithm combines BCD and MM in a nested structure. The outer BCD loop alternates between optimizing $m_1$ (block 1) and jointly optimizing $(d_{r,1}, d_{r,2})$ (block 2). For block 2, the inner MM loop iteratively minimizes the convex surrogate function until convergence. The majorization chain ensures that each MM iteration monotonically decreases the objective function, while the BCD framework guarantees convergence to a stationary point. This nested BCD-MM approach significantly reduces computational complexity while maintaining solution quality.\vspace{-5pt}
\section{Numerical Simulations}\vspace{-3pt}
In this section, we evaluate our analytical results and demonstrate the advantages of the proposed method through numerical benchmark comparisons. Unless otherwise specified, the simulations are following: the normalized transmit power is set to \( p = 1 \, \text{W} \), and the total available block length is \( M = 1000 \). The packet size of the confidence message is \( d_m = 20 \). We set the reliability and leakage thresholds to $\varepsilon_{ab}^{\max}=1-\varepsilon_{ae}^{\min}=0.5 $, $ \varepsilon_{ba}^{\max}=1-\varepsilon_{be}^{\min}=0.5$.
\begin{figure}[t!]
    \centering
    \includegraphics[width=0.8\linewidth,trim = 0 0 0 50]{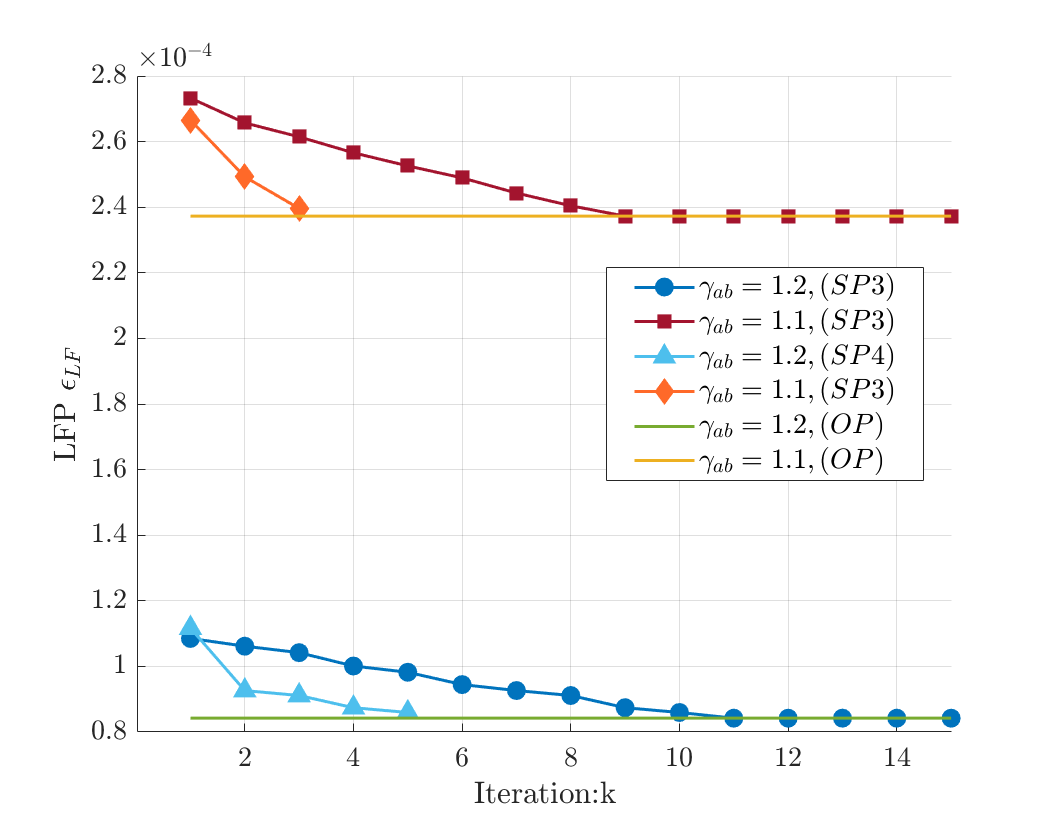}
    \caption{Obtained LFP $\varepsilon^{(k)}_{\text{LF}}$ in each $k$-th iteration for solving (SP2) with iterative search and under forward variant SNR of Bob $\varepsilon_{ab} = \{1, 1.2\}$. Moreover, the globally optimal results $\varepsilon^*_{\text{LF}}$ obtained with (OP) via integer programming are also shown as a benchmark.}
    \label{fig:fangzhen1}
\end{figure}

As discussed in Section~\ref{sec:system analyis}, (SP3) must be solved via iterative search. We plot the $\varepsilon_{LF}$ in each iteration \( k\)in Fig.~\ref{fig:fangzhen1}. To evaluate the algorithm's performance, we compare its results with those obtained by exhaustive search for (OP). Since the results from exhaustive search are guaranteed to be globally optimal, we can treat them as the lower bound. We can observe that the iterative results converge at a sublinear rate, which verifies our analysis. Meanwhile, it can obtain near-optimal solutions under various settings.We also plotted the iteration results obtained by (SP4) in Fig.~\ref{fig:fangzhen1} to verify the performance of the convex approximation algorithm in (SP5). It is easy to observe from the figure that the (SP4) algorithm can also work under various scenario settings, and it requires fewer iterations to achieve convergence.
\begin{figure}[t!]
    \centering
    \includegraphics[width=0.8\linewidth,trim = 0 0 0 50]{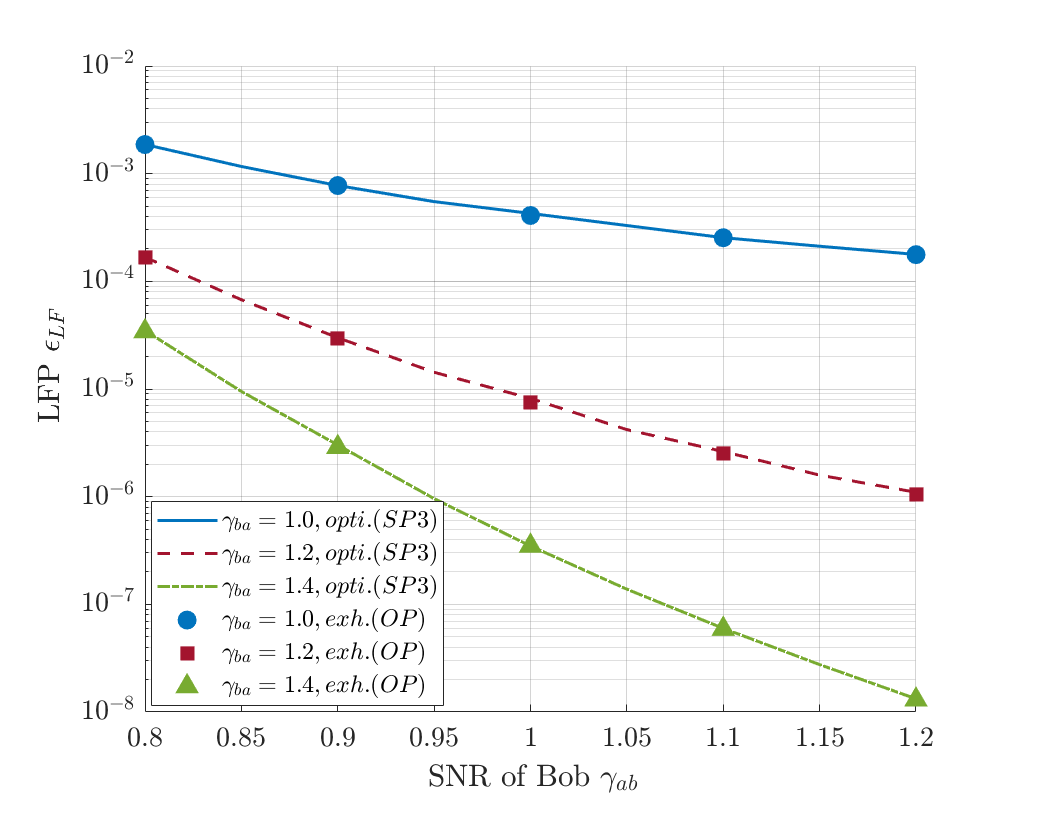}
    \caption{Minimized LFP $\varepsilon_{LF}$ against Bob’s SNR $\gamma_{ba}$ under various setups of Alice’s SNR $\gamma_{ab}$. The results obtained by solving (SP3), solving (OP), as well as with the assumption of the IBL are shown.}
    \label{fig:fangzhen2}
\end{figure}

Fig.~\ref{fig:fangzhen2} plots the relationship between the minimum LFP $\varepsilon_{LF}$ and Alice's SNR $\gamma_{{ab}}$ under different SNR settings of Bob $\gamma_{{ba}}$, aiming to evaluate the performance of (SP3). The iterative results of the BCD method (SP3) are labelled as opti., are compared with the results obtained by (OP), which are labelled as exh. The results show that the BCD method closely associated with the global optimal solution and performs well under various environmental conditions. This proves that (SP3) have a strong guiding design ability for the practical application.

\begin{figure}[t!]
    \centering
    \includegraphics[width=0.8\linewidth,trim = 0 0 0 50]{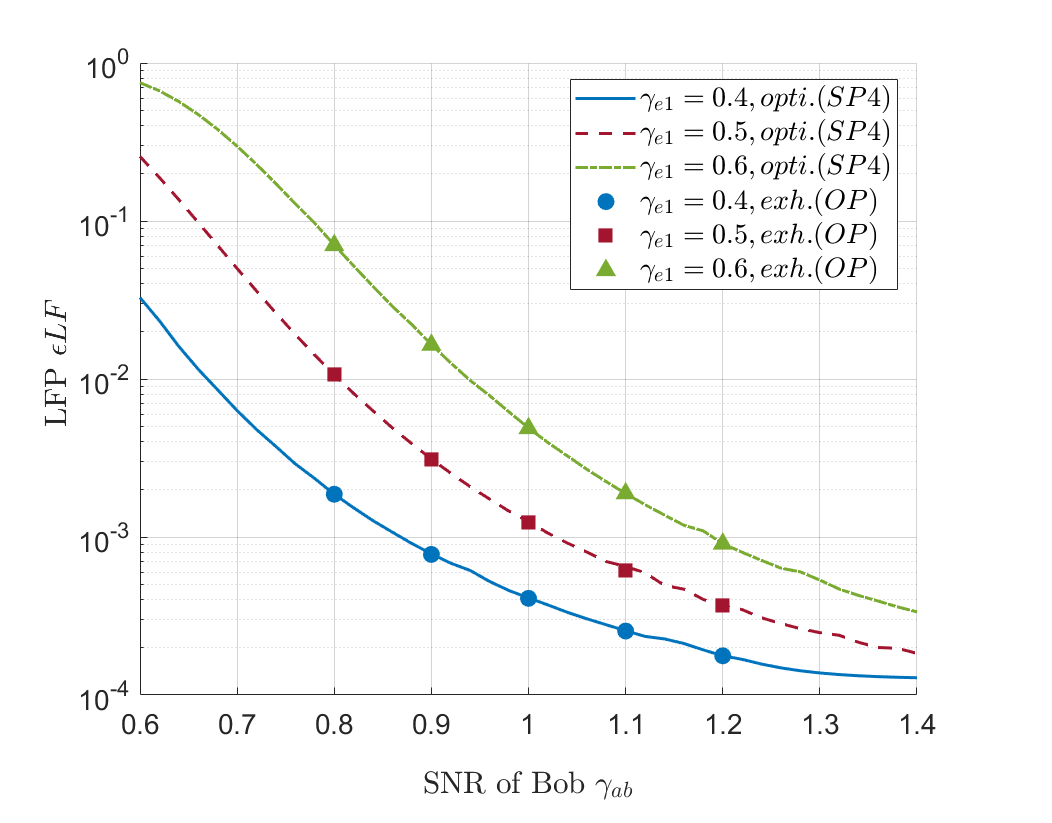}
    \caption{Minimized LFP $\varepsilon_{LF}$ against Bob’s SNR $\gamma_{ba}$ under various setups of Eve’s SNR $\gamma_{ae}$. The results obtained by solving (SP4), solving (OP), as well as with the assumption of the IBL are shown.}
    \label{fig:fangzhen3}
\end{figure}

Also, Fig.~\ref{fig:fangzhen3} plots the relationship between the LFP $\varepsilon_{LF}$ and Eve's SNR $\gamma_{ae}$ under different SNR settings of Bob $\gamma_{ab}$, which is used to evaluate the performance of (SP4) and (SP5). For specific points, (SP5) adopts the convex approximation MM algorithm  for iterative search. The results obtained by (SP4) are marked as opti.(SP4). Similarly, these results are compared with those obtained by (OP). It is easy to observe that (SP4) also exhibits eminent performance in searching for the global optimal solution, and is not falling into local optimal solutions. MM algorithm  of (SP4) maintains stable under various environmental parameter settings, providing application capabilities for practical use.

\begin{figure}[t!]
    \centering
    \includegraphics[width=0.8\linewidth,trim = 0 0 0 50]{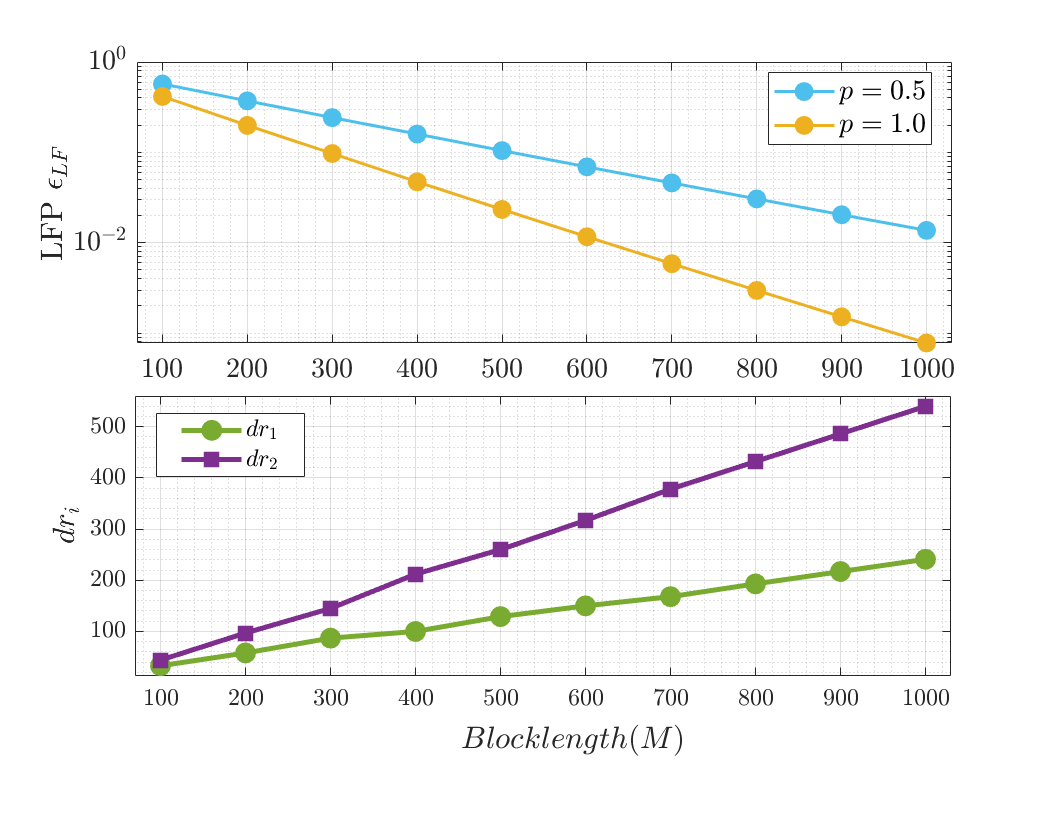}
    \caption{Minimized LFP $\varepsilon_{LF}$ and its corresponding redundant bits $d_{ri}$ against the available blocklength M under various setups of transmit power p.}
    \label{fig:fangzhen4}
\end{figure}

As shown in Fig.~\ref{fig:fangzhen4}, an increase in \( M \) reduces the LFP $\varepsilon_{LF}$, while it causes a compensating increase in $d_{ri}$. This indicates that an improvement in reliable safety is inevitably accompanied by a higher load size \( d_r \), which is consistent with our theoretical analysis. In essence, the enhancement of the joint secure-reliability is achieved by sacrificing a slight degree of reliability to gain a significant improvement in safety, which aligns with pracitical application scenarios. With \(M\) increasing, this overall improvement is further strengthened.\vspace{-7pt}\par
\section{Conclusion}
\label{sec:conclusion}
In this work, we investigate the impact of redundancy on PLS performance and FBL codes. We successively propose two efficient optimization methods for redundancy bits and blocklength. After formulating the certificate optimization problem aimed at minimizing the LFP, we first analyze and eliminate one variable, converting the problem into separate relaxed convex optimization problems in three variables. Subsequently,by proofing the joint convex property of two variables, we solve the problem using a convex approximation method. Both the two methods are applicable to various environmental scenario settings, revealing the key trade-off between reliability and safety under practical applications, and are consistent with theoretical analyses.
\vspace{-7pt}\par
\bibliographystyle{IEEEtran}
\bibliography{main}

\end{document}